\documentclass[12pt]{amsart}

\usepackage[T1]{fontenc}
\usepackage{amsthm, amsfonts, mathrsfs}
\usepackage{amssymb, amsmath, theoremref}
\usepackage{times, txfonts}

\theoremstyle{plain}
\newtheorem{claim}{\sc Claim}[section]
\newtheorem{corollary}[claim]{\sc Corollary}
\newtheorem{lemma}[claim]{\sc lemma}
\newtheorem{proposition}[claim]{\sc Proposition}
\newtheorem{theorem}[claim]{\sc Theorem}

\theoremstyle{definition}
\newtheorem{definition}[claim]{\sc Definition}
\newtheorem{conjecture}[claim]{\sc Conjecture}

\theoremstyle{remark}
\newtheorem{remark}[claim]{\sc Remark}

\begin{document}
\title{An extension of the Dirac theory of constraints}
\author{Larry Bates and J\k{e}drzej \'Sniatycki}
\date{\today }

\begin{abstract}
Constructions introduced by Dirac for singular Lagrangians are extended and
reinterpreted to cover cases when kernel distributions are either
nonintegrable or of nonconstant rank, and constraint sets need not be closed.
\end{abstract}

\maketitle

\noindent Apparently motivated in part by examples in field theory such as
general relativity, Dirac initiated the study of the canonical dynamics of
degenerate Lagrangians in finite dimensions \cite{dirac50}. In this
fundamental work, Dirac introduced several important notions. These include
the classification of constraints into different classes, the notion of
strong and weak equations, and what is now known as the Dirac bracket.
Implicit in his work was that the various sets he was considering were
manifolds, and this allowed him to conclude that certain constructions would
actually reduce dimension, and so if continued recursively, terminate.

In this paper we extend Dirac's considerations in several ways. In
particular, we give a new interpretation of the Dirac bracket. Instead of
merely using it to compute a constrained Poisson bracket, we repurpose it as
an operator that allows us to modify functions into first class functions.
Also considered is the case where the rank of the Legendre transformation is
not constant, and constraints are not defined by smooth functions. In this case the
Hamiltonian is not continuous and a regularity criterion is introduced in
order to have well defined equations of motion. Such notions are considered
in order to see what portion of the totality of all solutions of the
Euler-Lagrange equations can be constructed in the Hamiltonian formalism on
the cotangent bundle.

Since the Dirac constraint theory is of interest primarily because of its
applications, and hopeless to formulate in general for arbitrary
Lagrangians, this paper is structured as follows. First, an example where
everything can be computed `by hand' so to speak, and is of interest in and
of itself. Then some general discussion of the theory takes place, and then
the example is re-examined in light of the theory. This case most closely
resembles the Dirac theory as he described it in \cite{dirac50} and \cite{dirac64}. This approach is then repeated for a second example, only this
time the Lagrangian is not of constant rank. This time the theory requires
more radical changes, and will not look as familiar, but hopefully the
reader will agree with the authors opinion that such changes are pretty much
the minimum necessary to have a theory that can encompass these examples.

\section{A nonintegrable example (a)}

\label{example-part-1}

In order to avoid any possible confusion, the nonintegrability referred to
in the heading is the nonintegrability of the kernel distribution of the
Legendre transformation. It has nothing to do with the integrability or
nonintegrability of the dynamical system defined by the Lagrangian. To
start, set 
\begin{equation*}
\psi = dz-y\,dx.
\end{equation*}
Observe that $\psi\wedge\mathop{\! \, \rm d \!}\nolimits \psi\neq0$, so the
distribution defined by the kernel of $\psi$ is not integrable. Define a
degenerate metric by $g=\psi\otimes\psi$, and use this to give a `kinetic
energy' Lagrangian $l$ as 
\begin{equation*}
l=\frac12(y^2\dot{x}^2-2y\dot{x}\dot{z}+\dot{z}^2).
\end{equation*}

\subsection{The Euler-Lagrange equations}

The Euler-Lagrange equations given by the Lagrangian are 
\begin{align*}
x: \quad & \frac{d}{dt}(y^2\dot{x} -y\dot{z}) = 0, \\
y: \quad & \frac{d}{dt}(0) - (y\dot{x}^2 -\dot{x}\dot{z}) = 0, \\
z: \quad & \frac{d}{dt}(-y\dot{x} + \dot{z}) = 0.
\end{align*}
It follows from these equations that either $\dot{x}=0$ or $y\dot{x}-\dot{z}
=0$. Choosing the condition $y\dot{x}-\dot{z}=0$ yields a family of
solutions that contain two arbitrary functions and an arbitrary constant: 
\begin{equation*}
(x(t),y(t),z(t)) = \left(f(t), g(t), \int^t f^{\prime}(s)g(s) \,ds +c \right)
\end{equation*}
where $f(t)$ and $g(t)$ are arbitrary smooth function of $t$, and $c$ is a
constant.

If instead we choose $\dot{x}=0$, then $y\dot{z}=\mathrm{constant}$ and $
\dot{z}=\mathrm{constant}$. This implies that there are solutions of the
form 
\begin{equation*}
(x(t),y(t),z(t)) = \left(a, b, c +dt \right)
\end{equation*}
for constants $a,b,c,d$.

It is important to note that solutions of the first type with arbitrary
functions have the property that if $x(t)$ and $y(t)$ are constant, then $z(t)$ is a constant as well. This means that solutions of the second type
are not a special case of the first type for any choice of $f(t)$ and $g(t)$
and hence are distinct.

Solutions of the first type have the following constraints on initial
conditions: 
\begin{equation*}
x(0)=f(0),\quad \dot{x}(0)=\dot{f}(0),\quad y(0)=g(0),\quad \dot{y}(0)=\dot{g
}(0)
\end{equation*}
and all of these values may be chosen freely and independently. The initial
value of $z$ is $z(0)=c$, and $\dot{z}=\dot{x}y =\dot{f}g$ so, in particular 
$\dot{z}(0)=y(0)\dot{x}(0)$. In other words, there is a codimension one
constraint on our initial conditions that is given by the five dimensional
manifold defined by the equation $l=0$. In light of this, it is not
surprising that the Euler-Lagrange equations of the semidefinite Lagrangian
have solutions with the following property.

\begin{theorem}
Any two points in the submanifold $\mathscr{M}$ defined by $l=0$ may be
joined by a solution of the Euler-Lagrange equations.
\end{theorem}

\begin{proof}
Let the initial point in $\mathscr{M}$ be $(x_0,y_0,z_0,\dot{x}_0,\dot{y}_0, 
\dot{z}_0)$, with $\dot{z}_0 =y_0\dot{x}_0$. We claim that there exists a
curve $\gamma(t):[0,1]\rightarrow{\mathbb{R}}^2:t\rightarrow (f(t),g(t))$ in
the $x$-$y$ plane with the following properties.

\begin{enumerate}
\item The constraint on initial and final values: $\gamma(0)=\gamma(1)=(x_0,y_0)$.

\item The constraint on initial velocities: $f^{\prime}(0)=\dot{x}_0$ and $g^{\prime}(0)=\dot{y}_0$.

\item Third is that the integral $\int^t_0 f^{\prime}(s)g(s)\,ds =0$. This
is easy to achieve because it is a parametrization of the line integral $\int_{\gamma} y\,dx$, and by Green's theorem, the integral represents the
oriented area enclosed by $\gamma$.

\item The derivatives $f^{\prime}(1)$ and $g^{\prime}(1)$ may be assigned
any values $(\dot{x}_1,\dot{y}_1)$ whatsoever.
\end{enumerate}

Given such a curve $\gamma$, consider the associated curve $\Gamma$ defined
by 
\begin{equation*}
\Gamma(t) := \left(f(t),g(t), z_0 + \int^t_0 f^{\prime}(s)g(s)\,ds\right).
\end{equation*}
Then $\Gamma(t)$ lies in $\mathscr{M}$ and 
\begin{equation*}
\Gamma(0)=(x_0,y_0,z_0,\dot{x}_0,\dot{y}_0,\dot{z}_0), \qquad
\Gamma(1)=(x_0,y_0,z_0,\dot{x}_1,\dot{y}_1,\dot{z}_1).
\end{equation*}
The freedom in assigning the values of the derivatives $f^{\prime}(1)$ and $g^{\prime}(1)$ implies that the point $\Gamma(1)$ may be any point of the
constraint set lying over the configuration space point $(x_0,y_0,z_0)$.

The next step is to argue that there is a curve that fixes the values $x_0$
and $y_0$ and changes $z_0$ to the desired value $z_1$. This is done by
changing the integral constraint $(3)$, which was zero, to be the difference 
$z_1-z_0$. Composing this curve with a curve of the previous type shows that
any two points in $\mathscr{M}$ with $x_0=x_1$ and $y_0=y_1$ may be
connected with a solution of the Euler-Lagrange equations. Finally, one need
only first connect $(x_0,y_0,z_0)$ to the point $(x_1,y_1,\tilde{z})$, not
worrying about the derivative values or the value of the intermediate point $\tilde{z}$ to conclude that the manifold $\mathscr{M}$ is connected by
solutions of the Euler-Lagrange equations.\footnote{Readers with a background in control theory are no doubt familiar with this
sort of argument.}
\end{proof}

\begin{remark}
The reader should not fail to observe that this sort of behaviour can not
happen in any regular Lagrangian system with more than one degree of
freedom, and so get a glimpse of why the study of degenerate systems is
interesting. The dimension of the `reachable set' of a point in a degenerate
Lagrangian system depends not only on the rank of the kernel distribution $D$, but also the rank of the derived flag $D, D+[D, D], \dots$ A related
problem is to understand how to count the number of independent `gauge
functions' that show up in such examples. For example, the Lagrangian 
\begin{equation*}
l=\frac{1}{2}((\dot{z}-y\dot{x})^2+(\dot{y}-w\dot{x})^2)
\end{equation*}
has a kernel distribution of rank two, but there are gauge-like solutions of
the Euler-Lagrange equations of the form 
\begin{align*}
w(t) &= f(t) \\
x(t) &= g(t) \\
y(t) &= \int^t f(s)g^{\prime}(s) \, ds \\
z(t) &= \int^t\int^s g^{\prime}(s)f(r)g^{\prime}(r) \,dr\,ds
\end{align*}
and the reachable set has dimension six.
\end{remark}

For solutions of the second type, we can have the initial values $x(0)$,$
y(0) $, $z(0)$ all be arbitrary, while the velocity constraints are $\dot{x}
(0)=\dot{y}(0)=0$, with $\dot{z}(0)$ is arbitrary.

It follows that the allowable initial conditions are the union of two
manifolds, one of dimension five, and the other of dimension four. The
picture to keep in mind is that in each tangent space, the allowable initial
conditions for solutions of the first type form a plane containing the
origin, and the allowable initial conditions for solutions of the second
type form a line through the origin (the $\dot{z}$ axis), and these are
everywhere transverse. Define the `initial data constraint' of the problem
to be the set of all points $p$ in the tangent bundle $TQ$ such that there
exists an interval $I$ about $t=0$ and a curve $\gamma:I\rightarrow TQ$ with 
$\gamma(0)=p$ such that $\gamma(t)$ is a solution of the Euler-Lagrange
equations. What is important here is that the initial data constraint set is 
\textit{not} a manifold.

\section{General considerations (a)}

\subsection{The primary constraint set}

A large part of the Dirac constraint theory consists of understanding what,
if any, differences exist between the Hamiltonian and Euler-Lagrange
descriptions of the dynamics. To this end, we consider the energy. But first
some notation.

The configuration space is denoted by $Q$, the tangent bundle by $TQ$, the
projection $\tau :TQ\rightarrow Q$, the cotangent bundle by $T^*Q$, with
projection $\pi :T^{\ast }Q\rightarrow Q$. The Lagrangian is denoted by $l$,
the Legendre transformation by $\mathscr{L}$, the fundamental one form on
the cotangent bundle by $\vartheta_0$, and the symplectic form $d\vartheta_0$
by $\omega$. Denote the pullback of the symplectic form $\mathscr{L}^*\omega$
by $\omega_l$.\footnote{Since we are primarily interested in the case when the Legendre
transformation is not a diffeomorphism, the form $\omega_l$ need not be
symplectic, even though it is always closed.} For each smooth function $f\in
C^{\infty }(T^{\ast }Q)$, the Hamiltonian vector field of $f$ is the vector
field $X_{f}$ on $T^{\ast }Q$ such that 
\begin{equation*}
X_{f} \mbox{$\, \rule{8pt}{.5pt}\rule{.5pt}{6pt}\, \, $} \omega = -df.
\end{equation*}
The ring $C^{\infty }(T^{\ast }Q)$ of smooth functions on $T^{\ast }Q$ has
the structure of a Poisson algebra with the Poisson bracket 
\begin{equation*}
\{f_{1},f_{2}\}=-X_{f_{1}}f_{2}=X_{f_{2}}f_{1}=-\omega (X_{f_{1}},X_{f_{2}}).
\end{equation*}
The Poisson bracket is antilinear, and satisfies both the Leibniz rule and
the Jacobi identity. Define the energy $e$ at $v\in TQ$ by 
\begin{equation*}
e(v) := \langle \mathscr{L}(v),v\rangle -l(v).
\end{equation*}

For a curve $\gamma :t\rightarrow \gamma (t)$ in $Q$, the first jet
extension of $\gamma $ is the curve $j^{1}\gamma $ in $TQ$ associating to
each $t$ the tangent vector $\dot{\gamma}(t)$ to $\gamma $ at the point $\gamma (t)$.
If the Legendre transformation $\mathscr{L}:TQ\rightarrow
T^{\ast }Q$ is a diffeomorphism, then \ we have a globally defined
Hamiltonian $h$ on $T^{\ast }Q$ such that $e=\mathscr{L}^{\ast }h$, and a
curve $\gamma :t\rightarrow \gamma (t)$ in $Q$ satisfies the Euler-Lagrange
equations if and only if it is the projection to $Q$ of an integral curve of
the Hamiltonian vector field $X_{h}$ of $h$. For singular Lagrangians, 
\textsc{gotay} \cite{gotay79} has proved the equivalence of the Lagrangian
dynamics and the Hamiltonian dynamics on the range of the Legendre
transformation under some additional regularity conditions. 

We begin with a proposition folklore credits to Cartan.\footnote{The likely reference would seem to be his lectures on invariant integrals, 
\cite{cartan}, but we are unable to find this statement there. A proof of
its generalization can be found in \textsc{\'sniatycki} \cite{sniatycki70}.}

\begin{proposition}
Let $t\rightarrow \gamma (t)$ be a curve in $Q$ that satisfies the
Euler-Lagrange equations. Then the curve $t\rightarrow \frac{d}{dt}
j^1\gamma (t)\ $ in $T(TQ)$ satisfies the energy equation 
\begin{equation*}
\frac{d}{dt}j^1\gamma (t)\mbox{$\, \rule{8pt}{.5pt}\rule{.5pt}{6pt}\, \,
$}\omega _{l}=-\mathop{\! \, \rm d \!}\nolimits e(j^1\gamma(t)).
\end{equation*}
\end{proposition}

\begin{definition}
\thlabel{defn-constraint-set} The \textit{primary constraint set} is the
range $\mathscr{P}=\mathscr{L}(TQ)$ of the Legendre transformation. Let $\epsilon>0$ and $\gamma:(-\epsilon,\epsilon)\rightarrow TQ$ a solution of
the Euler-Lagrange equations. The set $\mathscr{D}$ of all $v\in TQ$ of the
form $\gamma(0)$ is the \textit{initial data constraint}, and the \textit{constraint set} $\mathscr{C}$ is the image $\mathscr{L}(\mathscr{D})$ of the
initial data constraint under the Legendre transformation.
\end{definition}

For this section assume that the primary constraint set is a closed
submanifold of $TQ$ and the Legendre transformation defines a submersion of $TQ$ onto $\mathscr{P}$. In order to have any sort of Hamiltonian theory on
the constraint set we need to assume that we can push the energy function
over to the cotangent bundle and construct a Hamiltonian. A condition that
guarantees this may be formulated as follows.

\begin{proposition}
\thlabel{energy-proposition} If the fibres of $\mathscr{L}$ are path
connected, the energy function $e$ pushes forward to a function $h:\mathscr{P}\rightarrow {\mathbb{R}}$.
\end{proposition}

\begin{proof}
Since for each $q\in TQ$, and $p\in \mathscr{P}\cap T^*_qQ$, the fibre $\mathscr{L}^{-1}(p)\subset T_qQ$ is path connected, it follows that a vector
in $T_v\mathscr{L}^{-1}(p)$ is tangent to a line $t\rightarrow v+tw$ with
the property 
\begin{equation*}
\frac{d}{dt} \mathscr{L}(v+tw)|_{t=0} = 0.
\end{equation*}
Therefore 
\begin{align*}
\frac{d}{dt}\,e(v+tw)|_{t=0} &= \frac{d}{dt}\{\langle \mathscr{L}
(v+tw),v+tw\rangle - \mathscr{L}(v+tw)\}_{t=0}, \\
&= \langle\mathscr{L}(v), w\rangle - D\mathscr{L}_v\cdot w, \\
&= \langle\mathscr{L}(v), w\rangle - \langle\mathscr{L}(v), w\rangle, \\
&= 0.
\end{align*}
This implies that the energy $e$ is constant along fibres of $\mathscr{L}$.
The assumption that fibres are path connected implies that $e(v)$ depends on 
$v$ only through $p=\mathscr{L}(v)$. It follows that $e=\mathscr{L}^*h$ for
some function $h:\mathscr{P}\rightarrow {\mathbb{R}}$.
\end{proof}

Further assume (for this section) that the constraint set $\mathscr{C}$ is a
closed subset of the cotangent bundle $T^*Q$. A theorem of Whitney \cite{whitney34}, together with a partition of unity argument, guarantees that
the function $h$ can be extended to a smooth function on all of $T^*Q$,
which we continue to denote by $h$. The function $h$ is the Hamiltonian of
the theory and integral curves of the Hamiltonian vector field of $h$ are
solutions of Hamilton's equations. Denote the pullback of $\omega$ to $\mathscr{C}$ by $\omega_{\mathscr{C}}$.

\begin{theorem}
\thlabel{EL-proj-ham} Suppose that the first jet $j^1\gamma(t)$ of the curve 
$\gamma:{\mathbb{R}}\rightarrow Q$ satisfies the Euler-Lagrange equations.
Then $\mathscr{L}j^1\gamma(t)$ satisfies the equation 
\begin{equation*}
\mathscr{L}j^1\gamma(t)\mbox{$\, \rule{8pt}{.5pt}\rule{.5pt}{6pt}\, \, $}
(\omega_{\mathscr{C}} +dh_{\mathscr{C}})=0.
\end{equation*}
\end{theorem}

\begin{proof}
By Proposition \ref{energy-proposition}, the curve $j^{1}\gamma (t)$
satisfies the energy equation 
\begin{equation*}
\frac{d}{dt}{j}^{1}\gamma (t) {
\mbox{$ \rule {5pt} {.5pt}\rule {.5pt} {6pt}
\, $}}\omega _{l}=-de(j^{1}\gamma (t)),
\end{equation*}
which is equivalent to 
\begin{equation*}
\mathscr{L}j^1\gamma(t)\mbox{$\, \rule{8pt}{.5pt}\rule{.5pt}{6pt}\, \, $}
(\omega_l +de)=0.
\end{equation*}
Since $\omega _{l}=\mathscr{L}^{\ast }\omega $, and $e=\mathscr{L}^{\ast }h_{
\mathscr{C}}$, it follows that 
\begin{equation*}
\frac{d}{dt}\mathscr{L}(j^{1}\gamma (t)) {\mbox{$ \rule {5pt} {.5pt}\rule
{.5pt} {6pt} \, $}} \omega = - dh(\mathscr{L}j^{1}\gamma (t)).
\end{equation*}
\end{proof}

Smooth functions that vanish on the primary constraint set are called
primary constraints. Denote by $\mathcal{P}$ the set of all primary
constraints. $\mathcal{P}$ is an associative ideal in $C^{\infty }(T^{\ast
}Q)$.

\begin{theorem}
If the first jet $j^{1}\gamma (t)$ \thinspace\ for a curve $\gamma :{\mathbb{R}}\rightarrow Q$ satisfies the Euler-Lagrange equations, then $\mathscr{L}
j^{1}\gamma (t)$ is an integral curve of a Hamiltonian vector field in $
T^{\ast }Q$ of the Hamiltonian $h+\lambda _{i}p_{i}$, where $p_{i}$ are
generators of the ideal $\mathcal{P}$ and $\lambda _{i}$ are Lagrange
multipliers.
\end{theorem}

\begin{proof}
The equation 
\begin{equation*}
\mathscr{L}j^{1}\gamma (t) {\mbox{$ \rule {5pt} {.5pt}\rule {.5pt} {6pt} \,
$}} (\omega _{\mathscr{C}}+dh_{\mathscr{C}})=0
\end{equation*}
is satisfied if and only if 
\begin{equation*}
\mathscr{L}j^{1}\gamma (t) {\mbox{$ \rule {5pt} {.5pt}\rule {.5pt} {6pt} \,
$}} (\omega +dh)+\lambda _{i}dp_{i}=0
\end{equation*}
for some $\lambda _{i}$ and $p_{i}$. Hence, 
\begin{equation*}
\mathscr{L}j^{1}\gamma (t) {\mbox{$ \rule {5pt} {.5pt}\rule {.5pt} {6pt} \,
$}} \omega =-dh-\lambda _{i}dp_{i},
\end{equation*}
which implies $\mathscr{L}j^{1}\gamma (t)$ is an integral curve of the
Hamiltonian vector field of $h+\lambda _{i}p_{i}$.
\end{proof}

\begin{remark}
The converse to this is not true. There are examples where a solution $\Gamma(t)$ of Hamilton's equations lying in the constraint set $\mathscr{C}$
can be projected to a curve $\pi\circ\Gamma(t)=\gamma(t)$ in the
configuration space, and the resulting curve $\gamma(t)$ does not satisfy
the Euler-Lagrange equations (see the discussion in the second section of
the nonintegrable example.) The reader may also profitably consult \textsc{gotay and nester} \cite{gotay-nester80} for a related discussion.
\end{remark}

\subsection{Secondary constraints}

\begin{definition}
A constraint function $c$ is a smooth function on the cotangent bundle $T^*Q$
that vanishes on the constraint set $\mathscr{C}$.\footnote{It is only necessary for the theory that things are defined in an open
neighbourhood of the constraint set, but for the sake of cleanliness of
exposition, we will assume constraint functions are globally defined.}
Denote the set of constraint functions by $\mathcal{C}$. The condition that
a function $c$ vanishes on the constraint set is commonly written $c\approx
0 $, and said to be a \textit{weak} equation.
\end{definition}

As in the case of primary constraints, the set of constraint functions is an
associative ideal, and assumed to be finitely generated. Note that $\mathcal{P}\subseteq\mathcal{C}$, because $\mathscr{C}\subset\mathscr{P}$.

Our aim is to describe the reduced phase space $\mathscr{R}$ of the theory,
together with its Poisson algebra. In Dirac's approach they are the
essential ingredients for quantization.

\begin{definition}
A function $f$ is said to be a \textit{first class function} if the Poisson
bracket of $f$ with any constraint function $c$ is also a constraint
function. That is, 
\begin{equation*}
\{ f,c\} \in \mathcal{C},
\end{equation*}
which is also written $\{f,c\}\approx 0$. Denote by $\mathcal{F}$ the set of
all first class functions.
\end{definition}

\begin{proposition}
The set $\mathcal{F}$ is a Poisson subalgebra of the Poisson algebra of
smooth functions $C^{\infty}(T^*Q)$.
\end{proposition}

\begin{proof}
Suppose that the functions $f_1$ and $f_2$ are in $\mathcal{F}$, and that
the function $c$ is in $\mathcal{C}$. Then the Poisson brackets 
\begin{align*}
\{f_1+f_2,c\} &= \{f_1,c\} + \{f_2,c\} \in \mathcal{C}, \\
\{f_1f_2,c\} &= f_1\{f_2,c\} +f_2\{f_1,c\} \in \mathcal{C},
\end{align*}
because $\mathcal{C}$ is an ideal in the associative algebra structure of $C^{\infty}(T^*Q)$.
\end{proof}

\begin{proposition}
The intersection $\mathcal{I}$ defined by 
\begin{equation*}
\mathcal{I} := \mathcal{C}\cap \mathcal{F} = \{ f\in\mathcal{C} \,|\,
\{f,c\}\in \mathcal{C} \text{ for all } c\in\mathcal{C}\,\}
\end{equation*}
is a maximal Poisson ideal in $\mathcal{C}$.
\end{proposition}

\begin{proof}
Since $\mathcal{C}$ is an ideal, and $\mathcal{F}$ a subalgebra in the
associative algebra structure of $C^{\infty}(T^*Q)$, it follows that the
intersection $\mathcal{I}=\mathcal{C}\cap\mathcal{F}$ is an ideal. Since for
every pair of functions $f_1$ and $f_2$ in $\mathcal{I}$, the Poisson
bracket $\{f_1,f_2\}\in \mathcal{C}$, it follows that if $c\in \mathcal{C}$,
then 
\begin{equation*}
\{\{f_1,f_2\},c\} = -\{\{f_2,c\},f_1\} - \{\{c,f_1\},f_2\} \in \mathcal{C}
\end{equation*}
as well. This implies that the bracket $\{f_1,f_2\}\in \mathcal{F}$, and
hence in $\mathcal{C}\cap \mathcal{F}=\mathcal{I}$. Moreover, if $f\in 
\mathcal{I}$ and $c_1,c_2\in \mathcal{C}$, then because the Poisson bracket
is a derivation in each slot, 
\begin{equation*}
\{c_1f,c_2\} = c_1\{f,c_2\} + \{c_1,c_2\}f \in \mathcal{C}.
\end{equation*}
This implies that $\mathcal{I}$ is a Poisson ideal in $\mathcal{C}$. To show
maximality of the ideal, consider $f\in \mathcal{C}$ such that $\{f,c\}\in 
\mathcal{C}$ for all $c\in \mathcal{C}$. It follows directly from the
definition that $f\in\mathcal{F}$, and therefore $f$ belongs to the
intersection $\mathcal{C}\cap\mathcal{F}=\mathcal{I}$.
\end{proof}

\begin{proposition}
The Poissson algebra structure on\, $\mathcal{F}$ generates a Poisson
algebra structure on $\mathcal{F}_{|\mathscr{C}}$, the restriction of\, $\mathcal{F}$ to the constraint set $\mathscr{C}$.
\end{proposition}

\begin{proof}
The quotient $\mathcal{F}/\mathcal{I}$ is an associative algebra because $\mathcal{I}$ is an associative ideal in $\mathcal{F}$. Bilinearity of the
bracket implies that if $f_1,f_2\in \mathcal{F}$, and $g_1,g_2\in \mathcal{I}
$, then the bracket 
\begin{equation*}
\{f_1+g_1, f_2+g_2\} = \{f_1,f_2\}+\{f_1,g_2\}+\{g_1,f_2\}+\{g_1,g_2\}.
\end{equation*}
Since the sum $\{f_1,g_2\}+\{g_1,f_2\}+\{g_1,g_2\}\in \mathcal{I}$, it
follows that a Poisson bracket on $\mathcal{F}_{|\mathscr{C}}$ is well
defined by the formula 
\begin{equation*}
\{f_{1|\mathscr{C}}, f_{2|\mathscr{C}} \} = \{f_1, f_2 \}_{|\mathscr{C}}.
\end{equation*}
\end{proof}

\begin{definition}
A function $f$ in the ideal $\mathcal{I}$ is called a \textit{first class
constraint}.
\end{definition}

We conclude this section with the description of the reduced phase space of
the system. For every first class function $f$, the Poisson bracket $\{f,c\}$
vanishes on $\mathscr{C}$ for each constraint $c$. Therefore, $f_{\mid \mathscr{C}}$ is constant along integral curves of $X_{c}$, the Hamiltonian
vector field of the constraint $c$, that are contained in $\mathscr{C}$.

\begin{definition}
Denote by $\mathfrak{P}$ the collection of all integral curves in $\mathscr{C}$ of the Hamiltonian vector fields $X_{c}$ for all constraint
functions $c\in \mathcal{C}.$ In other words, a curve $\gamma $ in $T^{\ast
}Q$ is in $\mathfrak{P}$ if $\gamma $ is an integral curve of the
Hamiltonian vector field of a constraint function, and $\gamma (t)\in 
\mathscr{C}$ for all $t$ in the domain of $\gamma$.
\end{definition}

Clearly, for each first class constraint $c$, all integral curves of the
Hamiltonian vector field $X_{c}$ through points in $\mathscr{C}$ are in $\mathfrak{P}$. Conversely, if for a constraint $c$, all integral curves of
the Hamiltonian vector field $X_{c}$ through points in $\mathscr{C}$ are in $\mathfrak{P}$, then $c$ is first class. However, there might be a constraint 
$c$ for which only some integral curves of $X_{c}$ through points in $\mathscr{C}$ are wholly contained in $\mathscr{C}.$

\begin{definition}
The points $p$ and $p^{\prime }$ in $\mathscr{C}$ are said to be \textit{equivalent}, written $p\sim p^{\prime }$, if $p$ can be connected to $p^{\prime }$ by a piecewise smooth curve in $\mathscr{C}$ with each smooth
piece contained in $\mathfrak{P}$. Clearly, $\sim $ is an equivalence
relation on $\mathscr{C}$. The space $\mathscr{R}$ of $\sim$-equivalence
classes in $\mathscr{C}$ is the \textit{reduced space} of the system.
\end{definition}

It follows from the definition above that every first class function $f$
restricted to $\mathscr{C}$ pushes down to a function on $\mathscr{R}$.

\begin{definition}
A constraint function is said to be \textit{second class} if it is a
constraint function and not first class.
\end{definition}

\begin{remark}
In the best of all worlds, second class constraints would appear in
canonically conjugate pairs, which would lead to symplectic submanifolds of $T^*Q$. However, the set of constraints can not necessarily be split nicely
into `independent sets' of first and second class constraints. See the
discussion in \textsc{gotay and nester} \cite{gotay-nester84}.
\end{remark}

In light of these results, the following should be a theorem following from
the fact that the constraint set is closed.

\begin{conjecture}
A function $f$ on the constraint set $\mathscr{C}$ extends to a first class
function if and only if it is constant along all the integral curves of $\mathfrak{P}$.
\end{conjecture}

In what follows we show that under additional conditions that if a
restriction $f|_{\mathscr{C}}$ of $f\in C^{\infty}(T^*Q)$ is constant on
curves in $\mathfrak{P}$ then $f|_{\mathscr{C}}$ extends to a smooth first
class function.

In order to ensure that a function $f$ is constant on curves in $\mathfrak{P}
$ Dirac considered a situation\footnote{Dirac did not say this explicitly, but he certainly appears to be aware of the situation, as a careful reading of \cite{dirac63} and \cite{dirac64} would show.} where the constraint ideal $\mathcal{C}$ is
generated by $n$ independent first class constraints $f_1,\dots,f_n$ and $k=2m$ second class constraints $s_1,\dots,s_k$ such that the matrix $S$ with
components $S_{ij}:=\{s_i,s_j\}$ of Poisson brackets is invertible with
inverse $A^{ij}$. The index convention for the inverse is $A^{jr}S_{rl}=\delta^j_l$.

\begin{definition}
The \textit{constraint modification map} that takes a function $f$ to the
modified function $f^*$ is 
\begin{equation*}
f^* = f - \{f,s_j\}A^{jl}s_l.
\end{equation*}
\end{definition}

Observe that $f$ and $f^*$ agree on the constraint set $\mathscr{C}$. It
also follows that we may view this as a (nonunique) way to extend functions
defined only on the constraint set to phase space. To do this, just take any
extension of $f$ to phase space, and then modify it with the constraint
modification map. The point of this modification is that it makes functions
originally defined only on the constraint set into first class functions.
Indeed, the following theorem holds.

\begin{theorem}
The modified function $f^*$ is a first class function.
\end{theorem}

\begin{proof}
From the definition of first class function, and assuming that the function
is constant on $\mathfrak{P}$ (which implies that the Poisson bracket $\{f^*,f_j\}=0$) it follows that we need only show that the Poisson bracket
of the modified function $f^*$ with a second class constraint function
vanishes. 
\begin{align*}
\{f^*,s_m\} & = \{f,s_m\} - \{\{f,s_k\}A^{kl},s_m\}s_l
-\{f,s_k\}A^{kl}\{s_l,s_m\} \\
& \phantom{=} \text{and restricting to the constraint set $\mathscr{C}$} \\
& = \{f,s_m\} - \{f,s_k\}\delta^k_m \\
& = 0.
\end{align*}
\end{proof}

Observing that it suffices to work locally, the following generalization
holds.

\begin{theorem}
\label{main}Assume that for every point $p\in \mathscr{C}$ there exists a
neighbourhood $U$ of $p\in T^{\ast }Q$ such that the restriction of the
constraint ideal $\mathcal{C}$ to $U$ is generated by $n_p$ independent
first class constraints $f_1,\dots,f_{n_p}$ and $m_{p}$ second class
constraints $s_1,\dots,s_{m_p}$ such that the $m_p\times m_p$ matrix of
Poisson brackets $\{s_a,s_b\}$ is nowhwere zero on $U$. Let $f\in C^{\infty
}(T^{\ast }Q)$. Assume that for each constraint function $c\in C$, $f_{\mid \mathscr{C}}$ is constant along all integral curves of $X_{c}$ that are
contained in $\mathscr{C}$. Then there exists a first class function $f^{\ast }\in C^{\infty }(T^{\ast }Q)$ such that $f_{\mid \mathscr{C}}=f_{\mid \mathscr{C}}^{\ast }$.
\end{theorem}

\begin{proof}
The proposition is established by using the construction given in the proof
of the previous proposition together with a partition of unity on $T^{\ast
}Q $ that is adapted to the neighbourhood $U$.
\end{proof}

\subsection{Reduced equations of motion}

Assuming that our conjecture is valid, the restrictions of first class
functions to $\mathscr{C}$ parametrize the reduced phase space. Hence, the
reduced equations of motion are completely determined by the evolution of
first class functions.

\begin{theorem}
\label{poisson-equation-of-motion} For every solution $\gamma (t)$ of the
Euler-Lagrange equations and each first class function $f$ 
\begin{equation*}
\frac{d}{dt}f(\mathscr{L}j^1\gamma (t))=\{f,h\}(\mathscr{L}j^1\gamma (t)).
\end{equation*}
\end{theorem}

\begin{proof}
According to a theorem above, $\mathscr{L}j^{1}\gamma (t)$ is an integral
curve of the Hamiltonian vector field of $h+\lambda _{i}p_{i}$, where $(p_{i})$ are generators of first class constraints. Therefore, 
\begin{eqnarray*}
\frac{d}{dt}f(\mathscr{L}j^1\gamma (t)) &=&\{f,h+\lambda _{i}p_{i}\}(\mathscr{L}j^1\gamma (t)) \\
&=&\{f,h\}(j^1\gamma (t))+\lambda _{i}\{f,p_{i}\}(\mathscr{L}j^1\gamma (t))
\\
&=&\{f,h\}(\mathscr{L}j^1\gamma (t))
\end{eqnarray*}
because $f$ is first class.
\end{proof}

Note that we there is no claim that the Hamiltonian is a first class
function, though it often is. Later we study an example in which the
Hamiltonian is not even continuous.

\subsection{Symmetries and constants of motion}

Consider an action 
\begin{equation*}
\phi :G\times Q\rightarrow Q:(g,q)\rightarrow \phi _{g}(q)=: g\cdot q
\end{equation*}
of a connected Lie group $G$ on $Q$. It lifts to an action $\phi ^{\prime }$
of $G$ on $TQ$ such that, for each $v\in TQ$ and $f\in C^{\infty }(Q)$, 
\begin{equation*}
\phi _{g}^{\prime }(v)(f)=v(\phi _{g}^{\ast }f),
\end{equation*}
where we have identified vectors in $TQ$ with derivations on $C^{\infty }(Q)$. Similarly, $\phi $ lifts to an action $\tilde{\phi}$ of $G$ on $T^{\ast }Q$
such that, for each $q\in Q$, $p\in T_{q}^{\ast }Q$ and $v\in T_{q}Q,$ 
\begin{equation*}
\left\langle \tilde{\phi}_{g}(p),v\right\rangle =\left\langle p,\phi
_{g^{-1}}^{^{\prime }}(v)\right\rangle.
\end{equation*}

The action $\tilde{\phi}$ of $G$ on $T^{\ast }Q$ preserves the canonical one
form $\vartheta _{0}$ and the symplectic form $\omega =d\vartheta _{0}.$ It
has an equivariant momentum map $j$ from $T^{\ast }Q$ to the dual $\mathfrak{g}^{\ast }$ of the Lie algebra $\mathfrak{g}$ of $G$ such that for $\xi \in 
\mathfrak{g}$, the action on $T^{\ast }Q$ of the one-parameter subgroup $\exp t\xi $ of $G$ is given by translation along integral curves of the
Hamiltonian vector field $X_{j_{\xi }}$, where 
\begin{equation*}
j_{\xi }=\left\langle j, \xi \right\rangle =\langle \vartheta _{0},
X_{j_{\xi }}\rangle .
\end{equation*}

\begin{definition}
The group $G$ is a \textit{symmetry group} of the Lagrangian $l$ if $\phi
_{g}^{\prime \ast }l=l$ for each $g\in G.$
\end{definition}

If $G$ is a symmetry group of $l$, then the Legendre transformation $\mathscr{L}$ intertwines the actions of $G$ on $TQ$ and on $T^{\ast }Q$ (a
proof is in the notes.) In particular, the action $\tilde{\phi}$ of $G$ on $T^{\ast }Q$ preserves the primary constraint set $\mathscr{P}=\mathscr{L}
(TQ) $. Since the Hamiltonian $h$ is defined on $\mathscr{P}$ in terms of
the Lagrangian $l$ and the action of $G$ preserves both $l$ and $\mathscr{P}$, it follows that the action $\tilde{\phi}$ preserves $h$. However, $\tilde{\phi}$ need not preserve the extension of $h$ off of $\mathscr{P}$. By the
first Noether theorem \cite{noether}, for each $\xi $ in the Lie algebra $\mathfrak{g}$ of $G$, the function $\mathscr{L}^{\ast}j_{\xi }$ is constant
along solutions of the Euler-Lagrange equations.

\begin{theorem}
If $G$ is a symmetry group of the Lagrangian $l$, then the action $\phi
^{\prime }$ of $G$ on $TQ$ preserves the initial data set $\mathscr{D}$ in $TQ$ and the action $\tilde{\phi}$ on $T^{\ast }Q$ preserves the constraint
set $\mathscr{C}.$
\end{theorem}

\begin{proof}
By definition, $v\in TQ$ is in the initial data set if there exists a finite
time solution $\gamma $ of the Euler-Lagrange equation such that $v=\dot{\gamma}(t)=j^1\gamma (t)$. Since the action $\phi $ of $G$ on $Q$ maps
solutions $\gamma$ of the Euler-Lagrange equations to solutions $\phi
_{g}\circ \gamma $ of the Euler-Lagrange equations, it follows that if $\gamma$ is a solution of the Euler-Lagrange equations such that $v=j^1\gamma
(t)$, then $\phi _{g}\circ \gamma $ is a solution of the Euler-Lagrange
equations and $\phi _{g}^{\prime}(v)=j^1(\phi _{g}\circ \gamma )(t)$. Hence, 
$\phi _{g}^{\prime }$ preserves the initial data set $\mathscr{D}$ in $TQ$.

Since the constraint set $\mathscr{C}$ is the image of the initial data set $\mathscr{D}$, and the Legendre transformation $\mathscr{L}$ intertwines the
actions $\phi ^{\prime }$ and $\tilde{\phi}$, it follows that the action $\tilde{\phi}$ also preserves the constraint set $\mathscr{C}$.
\end{proof}

\begin{corollary}
If $G$ is a symmetry group of the Lagrangian $l$, then the momentum $j_{\xi}\in C^{\infty }(T^{\ast }Q)$ is a first class function for every $\xi
\in \mathfrak{g}$.
\end{corollary}

\begin{proof}
For $\xi \in \mathfrak{g}$, the action on $T^{\ast }Q$ of the one-parameter
subgroup $\exp t\xi $ of $G$ is given by translation along the integral
curves of the Hamiltonian vector field $X_{j_{\xi }}$ of $j_{\xi }$. Since
this action preserves the constraint set $\mathscr{C}$, it follows that the
integral curves of $X_{j_{\xi }}$ through points of $\mathscr{C}$ are
contained in $\mathscr{C}$. This implies that $j_{\xi }$ is a first class
function.
\end{proof}

\begin{remark}
It should be noted that we may have an Hamiltonian action on $T^{\ast }Q$ of
a Lie group $G$ with an equivariant momentum map $j:T^{\ast }Q\rightarrow 
\mathfrak{g}^{\ast },$ which does not correspond to a symmetry of the
Lagrangian $l$. If this action preserves the Hamiltonian $h$, then the
momentum $j_{\xi }$ is a constant of motion. However, this action need not
preserve the constraint set $\mathscr{C}$, and the momentum $j_{\xi }$ need
not be a first class function.
\end{remark}

\section{A nonintegrable example (b)}

\subsection{The Legendre transformation}

Returning to our example, we find the canonical momenta are 
\begin{align*}
p_x & = y^2\dot{x} - y\dot{z}, \\
p_y & = 0, \\
p_z & = -y\dot{x} + \dot{z}.
\end{align*}
Observe that $p_x=-yp_z$. The image of the Legendre transformation is the
primary constraint set 
\begin{equation*}
\mathscr{P} := \{\,p_y=0, \,\, p_x+yp_z=0\,\}
\end{equation*}
and is a smooth manifold parametrized globally by $(x,y,z,p_z)$, which may
be thought of as a line bundle (with parameter $p_z$) over the zero section $(x,y,z, p_z=0)$ in the cotangent bundle.

From the first section the initial data set is 
\begin{equation*}
\mathscr{D}=\{(x,y,z,\dot{x},\dot{y},\dot{z}\mid \dot{x}(y\dot{x}-\dot{z}
)=0\},
\end{equation*}
which implies that on $\mathscr{D}$ either $y\dot{x}-\dot{z}=0$ or $\dot{x}
=0 $. The condition $y\dot{x}-\dot{z}=0$ implies $p_{z}=0$. On the other
hand, $\dot{x}=0$ does not imply any restrictions on the canonical momenta.
Therefore, the constraint set $\mathscr{C}$ coincides with the primary
constraint set $\mathscr{P};$ that is $\mathscr{C}=\mathscr{P}$.

It follows that $\mathscr{C}$ is a four dimensional submanifold of $T^{\ast
}Q$ parametrized by $(x,y,z,p_{z})$ as 
\begin{equation*}
\mathscr{C}=\{(x,y,z,-yp_{z},0,p_{z})\mid (x,y,z,p_{z})\in \mathbb{R}^{4}\}.
\end{equation*}

The functions $c_{1}=p_{y}$ and $c_{2}=p_{x}+yp_{z}$ generate $\mathcal{C}$
as an associative ideal in $C^{\infty }(T^{\ast }Q)$. Their Hamiltonian
vector fields are $X_{c_{1}}=\partial_y$ and $X_{c_{2}}=\partial_x+y
\partial_z-p_{z}\partial_{p_{y}}$. The Poisson bracket $\{c_{1},c_{2}\}=-p_{z}$, which means that the integral curves of $X_{c_{1}}$ and $X_{c_{2}}$ through points in $\mathscr{C}$ are contained in $\mathscr{C}$
provided $p_{z}=0$. The integral curves of $X_{c_{1}}$ and $X_{c_{2}}$ are 
\begin{eqnarray*}
\gamma _{1} &:&((x,y,z,p_{x},p_{y},p_{z}),t)\rightarrow
(x,y+t,z,p_{x},p_{y},p_{z}), \\
\gamma _{2} &:&((x,y,z,p_{x},p_{y},p_{z}),s)\rightarrow
(x+s,y,z+ys,p_{x},p_{y}-p_{z}s,p_{z}).
\end{eqnarray*}
The partition $\mathfrak{P}$ of $\mathscr{C}$ is given by 
\begin{equation*}
\mathfrak{P}=\{(x,y+t,z,0,0,0)\mid t\in \mathbb{R}\}\cup
\{(x+s,y,z+ys,0,0,0)\mid s\in \mathbb{R}\}.
\end{equation*}
Recall that functions in $C^{\infty }(\mathscr{C})$ push forward to
functions on the reduced space $\mathscr{R}$ of the system. A function $f\in
C^{\infty }(\mathscr{C})$ is constant along the curves in $\mathfrak{P}$ if 
\begin{equation*}
f((x+s,y+t,z+ys,0,0,0)=f((x,y,z,0,0,0)\text{ for }(s,t)\in \mathbb{R}^{2}.
\label{dynamical variable}
\end{equation*}

\begin{lemma}
A function $f\in C^{\infty }(\mathscr{C})$ that is constant along the curves
in $\mathfrak{P}$ is constant on the zero section.
\end{lemma}

\begin{proposition}
Every function $f\in C^{\infty }(\mathscr{C})$ that is constant along curves
in $\mathfrak{P}$ extends to a first class function in $C^{\infty }(T^{\ast
}Q)$. Conversely, every first class function is constant on the zero section.
\end{proposition}

\begin{proof}
Note that we cannot use Theorem \ref{main} directly, because its hypotheses
are not satisfied: the Poisson bracket $\{c_{1},c_{2}\}=-p_{z}$.

However, observing that $p_z$ parametrizes the fiber of $\mathscr{C}$ over
the zero section, $f|_{\mathscr{C}}$ may be written in the form 
\begin{equation*}
f(x,\dots,p_z):= k +p_zg(x,y,z,p_z)
\end{equation*}
with $k$ the constant value of $f$ on the zero section and $g$ an otherwise
arbitrary smooth function, $f$ extends to a smooth function on $T^*Q$ as 
\begin{equation*}
f(x,\dots,p_z)= k +p_zg(x,y,z,p_z)+ c_1f_1(x,\dots,p_z) +c_2f_2(x,\dots,p_z).
\end{equation*}
The conditions on $f$ to be a first class function are $\{f,c_1\}|_{
\mathscr{C}}=0$ and $\{f,c_2\}|_{\mathscr{C}}=0$. Explicitly, these are the
partial differential relations 
\begin{equation*}
f_1(x,y,z,-yp_z,0,p_z) = \frac{\partial g}{\partial x}(x,y,z,p_z)+ y\frac{\partial g}{\partial z}(x,y,z,p_z)
\end{equation*}
and 
\begin{equation*}
f_2(x,y,z,-yp_z,0,p_z) = -\frac{\partial g}{\partial y}(x,y,z,p_z).
\end{equation*}
There are no further obstructions in solving these equations for arbitrary $g $.
\end{proof}

An objection to the previous proof is that it is merely an existence proof,
and does not provide an explicit construction of a first class function that
extends the given one on the constraint set. This may be remedied as
follows. Set 
\begin{equation*}
f(x,y,z,p_z)=k+p_zg(x,y,z,p_z)
\end{equation*}
as before, where $g\in C^{\infty}(\mathscr{C})$. Since $\mathscr{C}$ is
closed in $T^{\ast }Q$, $g(x,y,z,p_{z})$ extends to a smooth function $f_1(x,y,z,p_{x},p_{y},p_{z})\in C^{\infty }(T^{\ast }Q).$ In other words, 
\begin{equation*}
g(x,y,z,p_{z})=f_{1}(x,y,z,-yp_{z},0,p_{z}).
\end{equation*}
Hence, $p_{z}f_{1}$ is an extension of $p_{z}g$ to $C^{\infty }(T^{\ast }Q)$, and $k+p_{z}f_{1}$ is an extension of $f=k+p_{z}g$ to $C^{\infty }(T^{\ast
}Q)$. In order to show that there exists a first class function on $T^{\ast
}Q$ that agrees with $k+p_{z}f_{1}$ on $\mathscr{C}$, we need only consider $p_{z}f_{1}$, since $k$ is a constant function, and is already first class.
\footnote{This actually characterizes all first class functions as a function is first
class implies that it is constant on the zero section and of the form $k+p_zg(x,y,z,p_z)$ on the constraint set.}

By construction, $p_{z}f_{1}$ is constant on $p_{z\mid \mathscr{C}}^{-1}(0)\subseteq $ $\mathscr{C}$. To find a first class function $f_{2}$
whose restriction to $\mathscr{C}$ coincides with the restriction to $\mathscr{C}$ of $p_{z}f_{1}$, we use Dirac's construction (valid when $p_{z}\neq 0$) and obtain 
\begin{equation*}
f_{2}^{0}=p_{z}f_{1}-\{p_{z}f_{1},c_{1}\}\frac{1}{\{c_{1},c_{2}\}}
c_{1}-\{p_{z}f_{1},c_{2}\}\frac{1}{\{c_{2},c_{1}\}}c_{2}.
\end{equation*}
It is clear that $f_{2}^{0}$ agrees with $p_{z}f_{1}$ on the complement of $p_{z\mid \mathscr{C}}^{-1}(0)\ $in $\mathscr{C}$. Moreover, the Poisson
brackets $\{f_{2}^{0},c_{1}\}$ and $\{f_{2}^{0},c_{2}\}$ vanish on the
complement of $p_{z\mid \mathscr{C}}^{-1}(0)\ $ in $\mathscr{C}$. However, 
\begin{eqnarray*}
f_{2}^{0} &=&p_{z}f_{1}-\{p_{z}f_{1},c_{1}\}\frac{1}{\{c_{1},c_{2}\}}
c_{1}-\{p_{z}f_{1},c_{2}\}\frac{1}{\{c_{2},c_{1}\}}c_{2} \\
&=&p_{z}f_{1}-\frac{\partial (p_{z}f_{1})}{\partial y}\,\frac{c_{1}}{-p_{z}}
-\left( \frac{\partial (p_{z}f_{1})}{\partial x}+y\frac{\partial
(p_{z}f_{1}) }{\partial z}-p_{z}\frac{\partial (p_{z}f_{1})}{\partial p_{y}}
\right) \frac{c_{2}}{p_{z}} \\
&=&p_{z}f_{1}-\frac{p_{z}\partial f_{1}}{\partial y}\,\frac{c_{1}}{-p_{z}}
-p_{z}\left( \frac{\partial f_{1}}{\partial x}+y\frac{\partial f_{1}}{
\partial z}-p_{z}\frac{\partial f_{1}}{\partial p_{y}}\right) \frac{c_{2}}{
p_{z}} \\
&=&p_{z}f_{1}+\frac{\partial f_{1}}{\partial y}\,c_{1}-\left( \frac{\partial
f_{1}}{\partial x}+y\frac{\partial f_{1}}{\partial z}-p_{z}\frac{\partial
f_{1}}{\partial p_{y}}\right) c_{2}.
\end{eqnarray*}

It follows that $f_{2}^{0}$ extends to a smooth function 
\begin{equation*}
f_{2}=p_{z}f_{1}+\frac{\partial f_{1}}{\partial y}\,c_{1}-\left( \frac{
\partial f_{1}}{\partial x}+y\frac{\partial f_{1}}{\partial z}-p_{z}\frac{
\partial f_{1}}{\partial p_{y}}\right) c_{2}
\end{equation*}
on $T^{\ast }Q$ that agrees with $p_{z}f_{1}$ on $\mathscr{C}$. The Poisson
brackets $\{f_{2},c_{1}\}$ and $\{f_{2},c_{1}\}$ are continuous and vanish
on the complement of $p_{z\mid \mathscr{C}}^{-1}(0)$ in $\mathscr{C}$. Since
the complement of $p_{z\mid \mathscr{C}}^{-1}(0)\ $ is dense in $\mathscr{C}$
, it follows that the the Poisson brackets $\{f_{2},c_{1}\}$ and $\{f_{2},c_{2}\}$ vanish on $\mathscr{C}$. Therefore, $f_{2}$ is a first
class function extending $p_{z}g\in C^{\infty }(\mathscr{C})$, and $k+f_{2}$
is a first class function extending $f=k+p_{z}g\in C^{\infty}(\mathscr{C})$.

\subsection{Equations of motion}

In the preceding section, we have shown that the constraint set $\mathscr{C}$
coincides with the range $\mathscr{P}$ of the Legendre transformation $\mathscr{L}$. Moreover, $\mathscr{C}$ is a four-dimensional submanifold of $T^{\ast }Q$ parametrized by $(x,y,z,p_{z})$ as 
\begin{equation*}
\mathscr{C}=\{(x,y,z,-yp_{z},0,p_{z})\in \mathbb{R}^{6}\mid (x,y,z,p_{z})\in 
\mathbb{R}^{4}\}.
\end{equation*}

The fundamental two form $\omega$ on $\mathscr{C}$ is 
\begin{equation*}
\begin{split}
\omega|_{\mathscr{C}} & = (-d\vartheta_0)|_{\mathscr{C}}, \\
& = -d(\vartheta_0|_{\mathscr{C}}), \\
& = -d(-yp_z\,dx +p_z\,dz), \\
& = - p_z \, dx\wedge dy - y dx \wedge dp_z + dz\wedge dp_z,
\end{split}
\end{equation*}
using $(x,y,z,p_z)$ as a global parametrization of the constraint set $\mathscr{C}$. Note that the four form $\omega\wedge\omega|_{\mathscr{C}}$ is 
\begin{equation*}
\omega\wedge\omega|_{\mathscr{C}} = -2p_z\, dx\wedge dy \wedge dz \wedge
dp_z,
\end{equation*}
and so is nondegenerate off of the zero section, and the constraint set $\mathscr{C}$ is not consistently oriented. The Poisson bracket $\{c_1,c_2\}
= -p_z$ is not constant on $\mathscr{C}$. Thus the constraints are \textit{second class}.

The energy function $e$ on $TQ$ pushes forward to a function $h_{\mathscr{C}} $ on $\mathscr{C}$, which extends to a function $h=\frac{1}{2}p_{z}^{2}$
on $T^{\ast}Q$. In other words, $e=\mathscr{L}^{\ast }h_{\mathscr{C}}$, and $h_{\mathscr{C}} = (\frac{1}{2}p_{z}^{2})_{\mid \mathscr{C}}$. Observe that $h $ as written is a first class function.

Define 
\begin{eqnarray*}
\mathscr{C}^{+} & :=&\{(x,y,z,p_{z})\in \mathscr{C}\mid p_{z}>0\}, \\
\mathscr{C}^{-} & :=&\{(x,y,z,p_{z})\in \mathscr{C}\mid p_{z}<0\}, \\
\mathscr{C}^{0} & :=&\{(x,y,z,p_{z})\in \mathscr{C}\mid p_{z}=0\}.
\end{eqnarray*}
It follows from the preceding that $\mathscr{C}^{+}$ and $\mathscr{C}^{-}$
are symplectic submanifolds of $T^{\ast }Q$ and $\mathscr{C}^{0}$ is
Lagrangian. Denote by $h^{+},h^{-},$ and $h^{0}$ the restriction of the
Hamiltonian $h$ to $\mathscr{C}^{+}$, $\mathscr{C}^{+}$, and $\mathscr{C}
^{0} $, respectively.

Now it is possible to see that the Euler-Lagrange equations of motion from
section \ref{example-part-1} can be obtained from the Hamiltonian equations.
On $\mathscr{C}^{\pm }$, the Hamiltonian vector field $X_{h^{\pm }}$ is
given by 
\begin{equation*}
X_{h^{\pm }} {\mbox{$ \rule {5pt} {.5pt}\rule {.5pt} {6pt} \, $}}
(p_{z}dx\wedge dy-ydp_{z}\wedge dx+dp_{z}\wedge dz)=-p_{z}dp_{z},
\end{equation*}
which implies 
\begin{equation*}
X_{h^{\pm}}=p_{z}\partial_z.
\end{equation*}
Therefore, 
\begin{equation*}
(x(t),y(t),z(t),p_{z}(t))=\left( a,b,c+td,d\right) ,
\end{equation*}
which agrees with our result in section \ref{example-part-1}. The equation
of motion $X_{h^0}\mbox{$\, \rule{8pt}{.5pt}\rule{.5pt}{6pt}\, \, $}\omega_{\mathscr{C}^0} =dh$ provides no restriction whatsoever on $X_{h^0}$ since
both $\omega_{\mathscr{C}^0}$ and $dh^0$ vanish. In this way the joining of
any two points in the zero section can be seen.

The Poisson bracket form of the the equations of motion has a somewhat
different character. Recall the characterization of first class functions,
and, in particular, that none of $x,y$ or $z$ are first class functions.
This means that we are not allowed to write an equation of motion for $x$ in
the form $\dot{x}=\{x,h\}$. Thus we are forced\footnote{Or, if you prefer, condemned.} to consider a more indirect approach. The
simplest substitutes would seem to be the ones that contain $x,y$ and $z$ as
linear factors of the first class functions 
\begin{align*}
f_1 &= k+p_zx-p_y, \\
f_2 &= k+p_zy-p_x-p_zy = k-p_x, \\
f_3 &= k+p_zz +p_yy.
\end{align*}
The equations of motion $\dot{f}=\{f,h\}$ imply 
\begin{equation*}
\dot{f}_1=0,\quad \dot{f}_2=0, \quad \dot{f}_3=2h.
\end{equation*}
Since the Poisson bracket is a derivation, and $f_2=p_zy-c_2$, etc., it
follows that 
\begin{align*}
\dot{f}_1 &= \dot{p}_zx + p_z\dot{x}-\dot{c}_1, \\
\dot{f}_2 &= \dot{p}_zy+p_z\dot{y}-\dot{c}_2, \\
\dot{f}_3 &= \dot{p}_zz + p_z\dot{z} +\dot{c}_1y +c_1\dot{y}.
\end{align*}
Since $\dot{c}_1=\dot{c}_2=0$, it follows that on the constraint set $\mathscr{C}$ 
\begin{align*}
p_z\dot{x} &= 0, \\
p_z\dot{y} &= 0, \\
p_z\dot{z} &= 2h.
\end{align*}
These equations are uniquely solvable for $\dot{x},\dot{y}$ and $\dot{z}$
only as long as $p_z\neq 0$. In this way, the evolution of $x,y$ and $z$ on $\mathscr{C}^0$ is seen to be arbitrary. However, and this is the interesting
point, there appears to be no way to see the velocity constraint $\dot{z}-y\dot{x}=0$ which enforces the special form of the evolution from section 1
on the cotangent bundle. This would appear to be linked to $p_x+yp_z$ not
being a first class function on all of $\mathscr{C}$.

\subsubsection{Conserved momenta and Noether's theorem}

By inspection, the Lagrangian is invariant under the symmetry group
generated by translations in $x$ and $z$. Usually, this would imply the two
independent conservation laws $p_x=$ constant and $p_z=$ constant. However,
one of the functions that defines the constraint set $\mathscr{C}$ is $p_x+yp_z=0$. Since there is a whole family of solutions of the
Euler-Lagrange equations that have $y(t)=g(t)$ with $g$ an arbitrary smooth
function, the only way that the constraint equation can hold is if it
implies that there are two constraints on the conserved values of the
momenta $p_x$ and $p_z$, and that is if they both simultaneously vanish. A
substitution show this to be the case. For solutions of the second type, it
is also true that $p_x$ and $p_z$ are constants of motion, but in this case
they do not need to both vanish, the momentum values just satisfy the
constraint equation $p_x+yp_z=0$.

\section{The case of nonconstant rank}

In the case of nonconstant rank, some of the definitions and constructions
of the received theory (as laid down in \cite{dirac50} or \cite{dirac64})
need to be modified. It must be stressed that such modifications to the
theory are forced from the examination of even the simplest examples. In
particular, some of the new difficulties to be dealt with are that
constraint sets need no longer be closed, and Hamiltonians need no longer be
continuous. 
A seemingly innocuous example that exhibits the typical difficulties of the
case of nonconstant rank is the following. Let the pseudometric $g$ be 
\begin{equation*}
g = y^2\, dx\otimes dx + x^2 dy \otimes dy
\end{equation*}
and the `kinetic energy' Lagrangian 
\begin{equation*}
l=\frac12(y^2\dot{x}^2+x^2\dot{y}^2).
\end{equation*}

The one-parameter group $\mathbb{R}$ acts on the configuration space $Q$ by 
\begin{equation*}
\phi_t(x,y) = (e^t x, e^{-t}y).
\end{equation*}
The derivative of the action at $t=0$ defines the vector field $X =
x\partial_x - y \partial_y$.

\begin{proposition}
The vector field $X:=x\partial_x - y \partial_y$ is a Killing field for the
pseudo-metric 
\begin{equation*}
g = y^2\, dx\otimes dx + x^2 dy \otimes dy.
\end{equation*}
\end{proposition}

\begin{proof}
From the formula for the Lie derivative of a two tensor 
\begin{equation*}
(\pounds _Xg)(Y,Z) = \pounds _X[g(Y,Z)] - g(\pounds _XY,Z) - g(Y,\pounds 
_XZ),
\end{equation*}
and the vector fields $X=x\partial_x-y\partial_y$, $Y=\partial_x$, $Z=\partial_y$ a calculation shows 
\begin{equation*}
(\pounds _Xg)_{xx} = (\pounds _Xg)_{xy} = (\pounds _Xg)_{yy} = 0 .
\end{equation*}
\end{proof}

From this it follows that the lifted vector field $\bar{X}$ acting on $TQ$
leaves the Lagrangian invariant, and thus there is a corresponding conserved
quantity $j:=xy^2\dot{x} - x^2y\dot{y}$.

A more complete discussion of solutions will occur on the Hamiltonian side.
The Euler-Lagrange equations are 
\begin{align*}
x: \quad \frac{d}{dt}(y^2\dot{x}) - x\dot{y}^2 &= 0, \\
y: \quad \frac{d}{dt}(x^2\dot{y}) - y\dot{x}^2 &= 0.
\end{align*}
If $y(t)\equiv 0$, a consequence is that there is no constraint on $x=f(t)$,
and vice versa. This implies that there is a $C^{\infty}$ smooth curve that
connects any two points on the subset of $TQ$ given by $xy=0$. In
particular, there is a smooth curve that connects $(1,0,\dot{x},0)$ with $(0,1,0,\dot{y})$ for arbitrary values of $\dot{x}$ and $\dot{y}$.

Observe that the set where the Lagrangian has rank $\le1$ is not a manifold,
but is completely path connected by solutions of the Euler-Lagrange
equations. Also note that straight lines of the form $y=kx$ are
(unparametrized) geodesics of the pseudometric.

\subsection{The Legendre transformation and reduction}

The Legendre transformation is 
\begin{equation*}
p_x = y^2\dot{x}, \qquad p_y = x^2 \dot{y}.
\end{equation*}
Thje image of the Legendre transformation is the primary constraint set $\mathscr{P}$, given by the two `distributional' constraints $p_y\delta(x)=0$
and $p_x\delta(y) =0$ (see the notes for further discussion.) Note that the primary constraint set is not closed.
Indeed, it is not even a manifold. This is not the complete description of
the constraint set $\mathscr{C}$. There is an additional constraint coming
from the condition that $\mathscr{C}=\mathscr{L}(\mathscr{D})$.

\begin{theorem}
For the Lagrangian $l= \frac12(y^2\dot{x}^2+x^2\dot{y}^2)$, the initial data
set $\mathscr{D}\neq TQ$.
\end{theorem}

\begin{proof}
The conserved angular momentum $j=xy(y\dot{x}-x\dot{y})$ implies that when
we search for a solution of the Euler-Lagrange equations with initial value 
\begin{equation*}
(x_0,y_0,\dot{x}_0,\dot{y}_0)=(1,0,0,1)
\end{equation*}
that $j=0$. However, this means that as soon as we move off of the $x$-axis,
that $xy\neq0$, so it must be the case that $y\dot{x}-x\dot{y}\equiv 0$.
This means that $d/dt(y/x)=0$ along the motion, or that $y=kx$ for some
constant $k$. This is a contradiction for $k\neq0$.
\end{proof}

It follows that the constraint set $\mathscr{C}\subset T^*Q$ is 
\begin{equation*}
\{(x,y,p_x,p_y)|xy\neq 0\}\cup\{(x,0,0,0)|x\in \mathbb{R}\} \cup
\{(0,y,0,0)|y\in \mathbb{R}\}.
\end{equation*}

The expression for the conserved quantity $j$ on the cotangent bundle is $j
= xp_x-yp_y$, and the Hamiltonian, defined as the push-forward of the energy
is 
\begin{equation*}
h = \frac12\left(\frac{p_x^2}{y^2}+ \frac{p_y^2}{x^2}\right)
\end{equation*}
on the open dense regular component $xy\neq 0$. The important point here is that there is no possible extension of the
Hamiltonian to all of phase space as a continuous function, despite the fact
that it is the push-forward of a smooth function by a smooth map. This means
that it is not immediately clear how to even write down Hamilton's
equations. However, in the dense open set where the Hamiltonian $h$ is
smooth, the Poisson bracket $\{j,h\}$ vanishes.

\begin{theorem}
The system with Hamiltonian 
\begin{equation*}
h = \frac12\left(\frac{p_x^2}{y^2}+ \frac{p_y^2}{x^2}\right)
\end{equation*}
is completely integrable. An additional Poisson commuting integral is
provided by $j = xp_x-yp_y$.
\end{theorem}

\subsubsection{Reduction}

Define $q:=xy/\sqrt{2}$, $p:=(p_x/y +p_y/x)/\sqrt{2}$, then the Poisson
bracket 
\begin{equation*}
\{q,p\}=1.
\end{equation*}
Furthermore, since the conserved momentum $j=xp_x-yp_y$, 
\begin{equation*}
\{q,j\} =0, \qquad \{p,j\}=0.
\end{equation*}
Since 
\begin{equation*}
\frac{p^2}{2} = \frac{1}{4}\left(\frac{p_x^2}{y^2}+\frac{p_y^2}{x^2}\right)
+ \frac{p_xp_y}{2xy}
\end{equation*}
and $j^2 = x^2p_x^2+y^2p_y^2 -2xyp_xp_y$, it follows that 
\begin{equation*}
\frac{j^2}{8q^2}= \frac{j^2}{4x^2y^2} = \frac{1}{4}\left(\frac{p_x^2}{y^2}+
\frac{p_y^2}{x^2}\right) - \frac{p_xp_y}{2xy}.
\end{equation*}
This implies that at $j=\mu$, the reduced Hamiltonian is 
\begin{equation*}
h_{\mu} = \frac{p^2}{2} + \frac{\mu^2}{8q^2},
\end{equation*}
and so the reduced equations of motion are 
\begin{equation*}
\dot{q}=\{q,h_{\mu}\}=p, \qquad \dot{p}:= \{p,h_{\mu}\} = -\frac{\mu^2}{4q^3}.
\end{equation*}

Solving for $p$ in the reduced Hamiltonian gives $p=\sqrt{2(h_{\mu}-\frac{\mu^2}{8q^2})}=\dot{q}$ and separating, integrating and solving for $q(t)$
yields 
\begin{equation*}
q(t) = \sqrt{2h_{\mu}(t+t_0)^2 + \frac{\mu^2}{8h_{\mu}}}.
\end{equation*}
From this expression it follows that $q(t)$ can be zero only if the value of
the momentum $\mu=0$. Setting $t_0=0$, $xp_x -yp_y =0$ or $xy^2\dot{x}-yx^2
\dot{y}=0$, factors as 
\begin{equation*}
xy(y\dot{x}-x\dot{y}) = 0,
\end{equation*}
so if $q\neq 0$, then $y\dot{x}-x\dot{y}=0$, or, multiplying by the
integrating factor $1/y^2$, 
\begin{equation*}
\frac{d}{dt}\left(\frac{x}{y}\right)=0.
\end{equation*}
This implies that $y=kx$ for some constant $k$. This agrees with the
previous observation that straight lines thorough the origin are
unparametrized geodesics. In particular, solving for $x(t)$ gives 
\begin{equation*}
x(t) = \sqrt[4]{\frac{4h_0}{k^2}} \sqrt{t}.
\end{equation*}
If $\mu\neq 0$, then the minimum value that $q$ obtains is 
\begin{equation*}
q_{min} = \frac{\mu}{\sqrt{8h_{\mu}}},
\end{equation*}
if things are positive, and we multiply by $-1$ to get a maximum if things
are negative.

A consequence of this analysis is that a solution of Hamilton's equations
starting with both $x$ and $y$ positive can not hit an axis in finite time
apart from the origin, and this is consistent with the description of the
constraint set $\mathscr{C}$.

\section{General considerations (b)}

The most important difference between the example discussed in this section
and the example introduced in section 1 is that both the primary constraint
set and the constraint set are no longer closed in $T^{\ast }Q$. and that
the Hamiltonian is not continuous. This is the main technical reason to
employ notions from the theory of differential spaces (see \cite{sniatycki13}).

\subsection{First class functions}

If $\mathscr{C}$ is a closed set, then 
\begin{equation*}
\mathscr{C}=\{p\in T^{\ast }Q\,|\,c(p)=0\text{\thinspace\ for all\thinspace }
c\in \mathcal{C}\},
\end{equation*}
where 
\begin{equation*}
\mathcal{C}:\{c\in C^{\infty }(T^{\ast }Q\mid c_{\mid \mathscr{C}}=0\}.
\end{equation*}
In this case, the smooth functions on $\mathscr{C}$ coincide with the space
of restrictions to $\mathscr{C}$ of smooth functions on $T^{\ast }Q$. This
is why the Dirac theory of constraints, as described in section 2, can be
formulated in terms of smooth functions on $T^{\ast }Q$. If $\mathscr{C}$ is
not closed, then the set of all constraint functions determines the closure $\overline{\mathscr{C}}$ 
\begin{equation*}
\overline{\mathscr{C}}=\{p\in T^{\ast }Q\,|\,c(p)=0{\text{\thinspace for
all\thinspace }}c\in \mathcal{C}\}.
\end{equation*}
In this case, if $f$ is such that the Poisson bracket of $f$ with a
constraint $c$ is a constraint, the flow of the Hamiltonian vector field $X_{f}$ of $f$ preserves the closure $\overline{\mathscr{C}}$ of the
constraint set. However, the flow need not preserve the set $\mathscr{C}$.
This explains why we make the following

\begin{definition}
A function $f\in C^{\infty }(T^{\ast }Q)$ is \textit{first class} if the
local flow of the Hamiltonian vector field $X_{f}$ of $f$ preserves the
constraint set $\mathscr{C}$.\footnote{The local flow means that the entire integral curve lies either entirely in
or entirely out of the constraint set. Thus the flow is a local flow in our
sense only if the vector fields are not complete.}
\end{definition}

As in Section 2, we denote by $\mathcal{F}$ the first class functions on $T^{\ast }Q$.

\begin{proposition}
The space $\mathcal{F}$ of first class functions inherits from $C^{\infty
}(T^{\ast }Q)$ the structure of a Poisson algebra with the bracket given by
the Poisson bracket.
\end{proposition}

\begin{proof}
Denote by $\exp tX_{f}$ the local one-parameter group of local
diffeomorphisms of $T^{\ast }Q$ generated by the flow of the Hamiltonian
vector field $X_{f}$ of $f$. If $f$ and $g$ are first class, then for each $p\in T^{\ast }Q$ there is an $\epsilon >0$ such that $\exp (-tX_{g})\circ
\exp (-tX_{f})\circ \exp (tX_{g})\circ \exp (tX_{f})(p)\in \mathscr{C}$ for $|t|<\epsilon $. This implies that the flow of $X_{\{g,f\}}=[X_{f},X_{g}]$
preserves $\mathscr{C}$. Therefore $\{f,g\}\in \mathcal{F}$. 
\end{proof}

The restrictions to $\mathscr{C}$ of functions in $\mathcal{F}$ define a
differential structure on $\mathscr{C}$ as follows.

\begin{definition}
A function $f_{\mathscr{C}}:\mathscr{C}\rightarrow {\mathbb{R}}$ is \textit{smooth} if, for each $p\in \mathscr{C}$, there exists an open neighbourhood $U$ of $p\in T^{\ast }Q$ and a function $f\in \mathcal{F}$ such that the
restrictions of $f$ and $f_{\mathscr{C}}$ to $U\cap \mathscr{C}$ coincide.
The space of smooth functions on $\mathscr{C}$ is denoted $C^{\infty }(\mathscr{C})$.
\end{definition}

\begin{remark}\label{subring}
If $\mathscr{C}$ is a proper subset of $T^{\ast }Q$, the space of
restrictions to $\mathscr{C}$ of the first class functions is a proper
subring of the space of restrictions to $\mathscr{C}$ of functions in $C^{\infty }(T^{\ast }Q)$. In this case, the differential structure $C^{\infty }(\mathscr{C})$ is different from the differential structure on $\mathscr{C}$ induced by the inclusion map.
\end{remark}

\begin{proposition}
The space $C^{\infty }(\mathscr{C})$ of smooth functions on $\mathscr{C}$
inherits from $\mathcal{F}$ the structure of a Poisson algebra.
\end{proposition}

\begin{proof}
For $f_{\mathscr{C}}$ and $g_{\mathscr{C}}$ in $C^{\infty }(\mathscr{C})$,
we define their Poisson bracket $\{f_{\mathscr{C}},g_{\mathscr{C}}\}$ as
follows. Given $p\in \mathscr{C}$, let $U$ be an open neighbourhood of $p\in
T^{\ast }Q$ such that there exists $f,g\in \mathcal{F}$ satisfying
\begin{equation*}
f_{\mathscr{C}\mid U\cap \mathscr{C}}=f_{\mid U\cap \mathscr{C}\text{ }},
\text{ and }g_{\mathscr{C}\mid U\cap \mathscr{C}}=g_{\mid U\cap \mathscr{C}}.
\end{equation*}
The restriction of $\{f_{\mathscr{C}},g_{\mathscr{C}}\}$ to $U\cap 
\mathscr{C}$ is given by 
\begin{equation*}
\{f_{\mathscr{C}},g_{\mathscr{C}}\}_{\mid U\cap \mathscr{C}}=\{f,g\}_{\mid
U\cap \mathscr{C}}.
\end{equation*}
Since the Poisson bracket depends only on the first jets of functions, it
follows that $\{f_{\mathscr{C}},g_{\mathscr{C}}\}$ is well defined, and it
satisfies the Leibnitz rule and the Jacobi identity.
\end{proof}

Denote by $\mathfrak{X}_{\mathcal{F}}$ the family of Hamiltonian vector
fields $X_{f}$ for all first class functions; that is 
\begin{equation*}
\mathfrak{X}_{\mathcal{F}}:=\{X_{f}\mid f\in \mathcal{F}\}.
\end{equation*}
By a generalization of a theorem of Sussmann, \cite{sniatycki13}, orbits of a family of vector fields are immersed manifolds.  Furthermore, these orbits coincide with the orbits of the linear hull of the family, and this explains why in Remark \ref{subring}, we can talk about the subring of first class functions.  The orbits of $\mathfrak{X}_{\mathcal{F}}$ give rise to a singular foliation $\mathfrak{F}$ of $T^*Q$. Since all vector fields in $\mathfrak{X}_{\mathcal{F}}$ preserve $\mathscr{C}$, it follows that every orbit $O\in 
\mathfrak{F}$ is either in $\mathscr{C}$ or in the complement of $\mathscr{C}$. Hence, 
\begin{equation*}
\mathfrak{F}_{\mathscr{C}}=\{O\in \mathfrak{F}\mid O\subseteq \mathscr{C}\}
\end{equation*}
is a partition of $\mathscr{C}$ by smooth manifolds.

\begin{proposition}
\label{Poisson orbit}For each orbit $O\in \mathfrak{F}_{\mathscr{C}}$, the
space 
\begin{equation*}
\mathcal{F}_{\mid O}=\{f_{\mid O}\mid f\in \mathcal{F}\}
\end{equation*}
of restrictions to $O$ of functions in $\mathcal{F}$ inherits from $\mathcal{F}$ the structure of a Poisson algebra.
\end{proposition}

\begin{proof}
Let $O$ be an orbit of $\mathfrak{X}_{\mathcal{F}}$. For each $f\in \mathcal{F},$ the Hamiltonian vector field $X_{f}$ of $f$ is tangent to $O$ and it
induces a vector field $X_{f\mid O}$ on $O$. The pull back to $O$ of the
defining equation for the Hamiltonian vector field of $f$ gives 
\begin{equation*}
X_{f\mid O}
{\mbox{$ \rule {5pt} {.5pt}\rule {.5pt} {6pt} \, $}}\omega _{O}=-df_{\mid O}
\end{equation*}
where $\omega _{O}$ is the pull back of $\omega $ to $O$. Hence, $X_{f\mid
O} $ depends on $f$ through its restriction to $O$.

For every $g\in \mathcal{F}$, 
\begin{equation*}
\{f,g\}_{\mid O}=-(X_{f}g)_{\mid O}=-X_{f_O}g_{\mid O},
\end{equation*}
which implies that $\{f,g\}_{\mid O}$ depends on the restrictions of $f_{\mid O}$ and $g_{\mid O}$ to $O$.
Hence we define a Poisson bracket on $\mathcal{F}_{\mid O}$ by 
\begin{equation*}
\{f_{\mid O},g_{\mid O}\}=\{f,g\}_{\mid O}.
\end{equation*}
It is easy to see that the bracket $\{f_{\mid O},g_{\mid O}\}$ is well
defined, and satisfies all the properties required of a Poisson bracket.
\end{proof}

\subsection{Reduced phase space}

In Section 2, we defined the reduced phase space in terms of the equivalence
relation on $\mathscr{C}$ defined by $p\sim p^{\prime }$ if $p$ can be
connected to $p^{\prime }$ by a piecewise smooth curve in $\mathscr{C}$ such
that each smooth piece is an integral curve of the Hamiltonian vector field
of a first class constraint. Moreover, we conjectured that a function on the constraint set $\mathscr{C}$ extends to a first class
function if and only if it is constant along all the integral curves of the
Hamiltonian vector fields of first class constraint that are contained in $\mathscr{C}$.

Since in the case under consideration, the constraint set is not defined in
terms of constraint functions, we define an equivalence relation on $\mathscr{C}$ as $p\approx q$ if $f(p)=f(q)$ for all first class functions $f$. The \textit{reduced phase space }is the space $\mathscr{R}$ of $\approx $-equivalence classes in $\mathscr{C}$. Let $\rho :\mathscr{C}\rightarrow \mathscr{R}$, denote the canonical projection. The reduced phase space $\mathscr{R}$ inherits from $\mathscr{C}$ a differential structure 
\begin{equation*}
C^{\infty }(\mathscr{R})=\{f:\mathscr{R}\rightarrow \mathbb{R}\mid \rho
^{\ast }f\in C^{\infty }(\mathscr{C})\}.
\end{equation*}
The main challenge of the theory is to understand the geometric structure of 
$\mathscr{R}$ endowed with this differential structure.

\subsection{Reduced Hamiltonian dynamics}

Our aim is to write the reduced equations of motion in the Poisson bracket
form as 
\begin{equation*}
\dot{f}_{\mid \mathscr{C}}=\{f_{\mid \mathscr{C}},h_{\mathscr{C}}\}.
\end{equation*}
where $h_{\mathscr{C}}=h_{\mathscr{P}\mid \mathscr{C}}$ is the restriction
to $\mathscr{C}$ of the Hamiltonian $h_{\mathscr{P}}$, originally defined on
the primary constraint set $\mathscr{P}=\mathscr{L}(TQ)$ by the Legendre
transformation $\mathscr{L}:TQ\rightarrow T^{\ast }Q$ as follows 
\begin{equation*}
\mathscr{L}^{\ast }h_{\mathscr{P}}(v)=\left\langle \mathscr{L}(v)\mid
v\right\rangle -L(v).
\end{equation*}
In the regular theory, discussed in Section 2, we assumed that $h_{\mathscr{P}}$ can be extended to a smooth function $h$ on $T^{\ast }Q$.
However, in the case of non-constant rank, discussed in the preceeding
section, the Hamiltonian $h_{\mathscr{P}}$ is discontinuous, and it has no
smooth extension to $T^{\ast }Q$. Therefore, we need a different approach.

We know that $\mathscr{C}$ is singularly foliated by $\mathfrak{X}_{\mathcal{F}}$-orbits and that each $\mathfrak{X}_{\mathcal{F}}$-orbit $O$ is a
Poisson manifold; see Proposition \ref{Poisson orbit}. If for each $\mathfrak{X}_{\mathcal{F}}$-orbit $O$, the restriction of $h_{\mathscr{C}}$
to $O$ is smooth, then we could define $\{f_{\mid \mathscr{C}},h_{\mathscr{C}}\}$ orbit by orbit, 
\begin{equation}
\{f_{\mid \mathscr{C}},h_{\mathscr{C}}\}_{\mid O}=\{f_{\mid O},h_{\mathscr{C}\mid O}\}=-X_{f\mid O}h_{\mid O}\text{.}  \label{Poisson bracket C}
\end{equation}
Since the Hamiltonian is completely determined by the Lagrangian, we are
lead to the following definition.

\begin{definition}
The Lagrangian $l$ on $TQ$ is agreeable if the restriction of the
Hamiltonian $h_{\mathscr{C}}$ to every $\mathfrak{X}_{\mathcal{F}}$-orbit $O$
is smooth.
\end{definition}

Note that a Lagrangian $L$, leading to a closed constraint set $\mathscr{C}$, is always agreeable. The notion of agreeability allows the equations of
motion to formally look the same as the constant rank case.

\begin{conjecture}
For an agreeable Lagrangian $L$, the Hamiltonian equations of motion for
first class functions are given by 
\begin{equation}
\dot{f}_{\mid \mathscr{C}}=\{f_{\mid \mathscr{C}},h_{\mathscr{C}}\},
\label{Hamiltonian equations}
\end{equation}
where the Poisson bracket on the right hand side is defined by equation (\ref{Poisson bracket C}).
\end{conjecture}

\noindent In the next section, we show that this conjecture holds for the
example introduced in Section 4.

Since first class functions parametrize the reduced phase space $\mathscr{R}, $ solutions of equation (\ref{Hamiltonian equations}) gives the reduced
dynamics. Given a curve in $\mathscr{R}$, corresponding to the reduced
motion, in order to get unreduced motion, we need to lift this curve to a
curve in $\mathscr{C}$. However, there is no dynamical restriction on such a
lift. The arbitrariness of the lift is responsible for arbitrary functions
of time occuring in some types of solutions of the Euler-Lagrange equations.

\section{A nonconstant rank example (b)}

It remains to examine how these notions aid in understanding the nonconstant
rank example under consideration.

\subsection{First class functions}

Recall that the Legendre transformation in this case is 
\begin{equation*}
p_{x}=y^{2}\dot{x},\qquad p_{y}=x^{2}\dot{y}.
\end{equation*}
The image of the Legendre transformation is the primary constraint set 
\begin{equation*}
\mathscr{P}=(\{(x,p_{y})\in \mathbb{R}^{2}\mid x\neq 0\}\cup
\{(0,0)\})\times (\{(y,p_{x})\in \mathbb{R}^{2}\mid y\neq 0\}\cup \{(0,0)\}).
\end{equation*}
The constraint set is 
\begin{equation*}
\mathscr{C}=\{(x,y,p_{x},p_{y})\mid xy\neq 0\}\cup \{(x,0,0,0)\mid x\in 
\mathbb{R}\}\cup \{(0,y,0,0)\mid y\in \mathbb{R}\},
\end{equation*}

In this case first class functions have the following restrictions on their
derivatives.

\begin{lemma}
\label{first class}A function $f$ is first class if it satisfies the partial
derivative relations 
\begin{equation*}
\left. \frac{\partial f}{\partial p_{x}}\right\vert _{(0,y,0,0)}=0,\quad
\left. \frac{\partial f}{\partial p_{y}}\right\vert _{(0,y,0,0)}=0,~\text{
and~}\left. \frac{\partial f}{\partial y}\right\vert _{(0,y,0,0)}=0
\end{equation*}
as well as 
\begin{equation*}
\left. \frac{\partial f}{\partial p_{x}}\right\vert _{(x,0,0,0)}=0,\quad
\left. \frac{\partial f}{\partial p_{y}}\right\vert _{(x,0,0,0)}=0,~\text{
and~}\left. \frac{\partial f}{\partial x}\right\vert _{(x,0,0,0)}=0
\end{equation*}
for all $x,y\in \mathbb{R}$.
\end{lemma}

\begin{proof}
The Hamiltonian vector field $X_f$ of $f$ restricted to $\overline{\complement\mathscr{C}}$, the closure of the complement of the constraint
set $\mathscr{C}$, has no $\partial_x$ component when $x=0$, and no $\partial_y$ component when $y=0$. Thus the Hamiltonian function has the
given restrictions on its partial derivatives. On $p_x=p_y=xy=0$, the
Hamiltonian vector field must have the further restriction that it has
vanishing $\partial_{p_x}$ and $\partial_{p_y}$ components. This gives the
vanishing of the other derivatives of the function.
\end{proof}

It follows from Lemma \ref{first class} that the family $\mathfrak{X}_{\mathcal{F}}$ of Hamiltonian vector fields for all first class functions has
nine connected orbits: the four component open orbit  
\begin{equation*}
O=\{(x,y,p_{x},p_{y})\mid xy\neq 0\},
\end{equation*}
four one-dimensional orbits 
\begin{equation*}
\{(x,0,0,0)\mid x>0\}\text{, }\{(x,0,0,0)\mid x<0\}\text{, }\{(0,y,0,0)\mid
y>0\}\text{, }\{(0,y,0,0)\mid y<0\}
\end{equation*}
and the origin $\{(0,0,0,0)\}$. The open orbit $O$ is a symplectic manifold,
and the Poisson algebra $\mathcal{F}_{\mid O}$ is a Poisson subalgebra of $C^{\infty }(O)$. The Poisson algebras of the restrictions of first class
functions to other orbits have trivial Poisson brackets. 

The second consequence of Lemma \ref{first class} is that the equivalence
relation $p\sim q$ if $f(p)=f(q)$, for all first class functions $f$, is the
identity on $O$ and that all points in the union of the remaining orbits are
equivalent. Hence, the reduced space $\mathscr{R}$ is the union of $O$ and a
single point $\{\ast \}$,   
\begin{equation*}
\mathscr{R}=O\cup \{\ast \},
\end{equation*}
where 
\begin{equation*}
\{\ast \}=(\{(x,0,0,0)\mid x\in \mathbb{R}\}\cup \{(0,y,0,0)\mid y\in 
\mathbb{R}\})/\sim ,
\end{equation*}
and the projection map $\rho :\mathscr{C}\rightarrow \mathscr{R}$ is the
identity on $O.$

\subsection{Dynamics}

The Hamiltonian, defined on the primary constraint set, is

\begin{equation*}
h_{\mathscr{P}}(x,y,p_{x},p_{y})=\left\{ 
\begin{array}{ccc}
\frac{p_{x}^{2}}{2y^{2}}+\frac{p_{y}^{2}}{2x^{2}} & \text{on} & 
\{(x,y,p_{x},p_{y})\mid xy\neq 0\} \\ 
\frac{p_{x}^{2}}{2y^{2}} & \text{on} & \{(0,y,p_{x},0)\in \mathbb{R}^{4}\mid
y\neq 0\} \\ 
\frac{p_{y}^{2}}{2x^{2}} & \text{on } & \{(x,0,0,p_{y})\in \mathbb{R}
^{4}\mid x\neq 0\} \\ 
0 & \text{on} & \{(0,0,0,0)\}
\end{array}
\right. .
\end{equation*}
Its restriction to the constraint set $\mathscr{C}$ is  
\begin{equation*}
h_{\mathscr{C}}(x,y,p_{x},p_{y})=\left\{ 
\begin{array}{ccc}
\frac{p_{x}^{2}}{2y^{2}}+\frac{p_{y}^{2}}{2x^{2}} & \text{on} & 
\{(x,y,p_{x},p_{y})\mid xy\neq 0\} \\ 
0 & \text{on} & \{(0,y,0,0)\mid y>0\}\cup
\{(0,y,0,0)\mid y<0\} \\ 
0 & \text{on } & \{(x,0,0,0)\mid x>0\}\cup
\{(x,0,0,0)\mid x<0\} \\ 
0 & \text{on} & \{(0,0,0,0)\}
\end{array}
\right. .
\end{equation*}
It follows that the restriction of $h_{\mathscr{C}}$ to each $\mathfrak{X}_{\mathcal{F}}$-orbit is smooth. Therefore, the reduced equations of motion
are well defined. On the open orbit $O$ we have the usual Hamiltonian
equations of motion
\begin{equation}
\dot{f}_{\mid O}=\{f_{\mid O},h_{\mathscr{C}\mid O}\}  \label{O}
\end{equation}
and all first class functions are constant on the lower dimensional orbits. It is important to note that the same Poisson bracket form holds on the lower dimensional orbits, but because first class functions are constant on these orbits, this implies the presence of gauge-like solutions. 

In general, given a motion in the reduced phase space, we need to get the
lifted motion in the constraint set. This process is usually called
reconstruction. In our case, since the projection map $\rho :\mathscr{C}
\rightarrow \mathscr{R}$ is the identity on $O\ $and at the origin, equation
(\ref{O}) does not require reconstruction, and the origin is a fixed point.
On the other hand, the one-dimensional orbits give rise to solutions
involving arbitrary functions. In particular 
\begin{equation*}
\begin{array}{ccc}
\{(0,y,0,0)\mid y>0\} & \text{yields} & x(t)=0,~y(t)=f(t)>0 \\ 
\{(0,y,0,0)\mid y<0\} & \text{yields} & x(t)=0,~y(t)=f(t)<0 \\ 
\{(x,0,0,0)\mid x>0\} & \text{yields} & x(t)=f(t)>0,~y(t)=0 \\ 
\{(x,0,0,0)\mid x<0\} & \text{yields} & x(t)=f(t)<0,~y(t)=0
\end{array}
.
\end{equation*}
In this way we see the gauge like solutions of the Euler-Lagrange equations
when in the Hamiltonian formalism.

Some reflection shows that these examples are likely typical for the general
case. Thus it seems probable that the only portion of the totality of the
solutions of the Euler-Lagrange equations that can not be seen in the Dirac
constraint theory is the velocity relations in the gauge like solutions.

\section{Notes}

\begin{enumerate}
\item The reader will have noticed that discussion of the Dirac constraint algorithm is conspicuous by its absence.  This is because it is well presented elsewhere in the literature, for example in \cite{dirac50} or \cite{gotay-nester-hinds78}.

\item The equivariance of the Legendre transformation is seen as follows.
Let a Lie group $G$ act on the configuration space $Q$ by $q\rightarrow
\phi_g\cdot q$. Then the induced action on $v\in T_qQ$ is $v\rightarrow {\phi_g}_*\cdot v$, and the induced action on $p\in T^*_qQ$ is $p\rightarrow
\phi^*_{g^{-1}}\cdot p$. The Legendre transform for the Lagrangian $l$ is
defined by 
\begin{equation*}
\frac{d}{dt} l(v+tw)|_{t=0} = \langle \mathscr{L}(w),v\rangle.
\end{equation*}
Then 
\begin{equation*}
\frac{d}{dt}l({\phi_g}_*\cdot v +tw)|_{t=0} =\langle\mathscr{L}(w),{\phi_g}
_*\cdot v\rangle = \langle \phi^*_{g^{-1}}\cdot\mathscr{L}(w),v\rangle.
\end{equation*}
If the Lagrangian $l$ is $G$-invariant, then the left hand side equals 
\begin{equation*}
\frac{d}{dt}l(v+t{\phi_{g^{-1}}}_*\cdot w)|_{t=0} = \langle\mathscr{L}({\phi_{g^{-1}}}_*\cdot w), v\rangle.
\end{equation*}
This means that 
\begin{equation*}
\mathscr{L}({\phi_g}_* \cdot w) =\phi^*_{g^{-1}}\cdot \mathscr{L}(w),
\end{equation*}
which is to say that the Legendre transformation is a $G$-equivariant map if
the Lagrangian is a $G$-invariant function.

\item The case of nonconstant rank raises the issue of just what constitute
the dynamical variables of the theory. If we stick with the usual case that
the first class functions are the dynamical variables, this leads to the
unpalatable notion that in general the Hamiltonian is not a dynamical
variable. A possible way around this is to look at the class of all
functions that behave like the Hamiltonian did in the nonconstant rank
example. However, the authors are not convinced that this is the best
solution to the problem, and so we leave it to others to come up with a
satisfying definition.

\item Our constraint modification map bears an obvious similarity to the
Dirac bracket defined by 
\begin{equation*}
\{f,g\}^* := \{f,g\} - \{f,s_j\}A^{jk}\{s_k,g\}.
\end{equation*}
The relation between Dirac's bracket and the normal Poisson bracket is not
only do they both satisfy the Jacobi identity, but that they are weakly
equal via the constraint modification map: 
\begin{equation*}
\{f^*,g^*\}\approx \{f,g\}^*
\end{equation*}
where the left hand side is the usual bracket of modified functions and the
right hand side is the Dirac bracket of unmodified functions.

\item It is a useful notation to write the constraints in a `distributional' form even though we do not employ the theory of distributions in any meaningful way.  However, it should be noted that this notation is quite convenient for calculating various quantites such as the conditions for functions to be first class in the nonconstant rank case.   

\item The problems raised by nonconstant rank and the attendant difficulties
with evolution equations have been used by physicists to exclude certain
field theories from consideration. Further discussion of these ideas may be
found in \textsc{chen et al} \cite{chen-nester-yo}.
\end{enumerate}

\vspace{20pt}

\noindent Larry M. Bates \newline
Department of Mathematics \newline
University of Calgary \newline
Calgary, Alberta \newline
Canada T2N 1N4 \newline
bates@ucalgary.ca

\vspace{20pt}

\noindent J\k{e}drzej \'Sniatycki \newline
Department of Mathematics \newline
University of Calgary \newline
Calgary, Alberta \newline
Canada T2N 1N4 \newline
sniatyck@ucalgary.ca

\end{document}